\documentclass[11pt]{article}
\usepackage{a4wide}

\usepackage{amsmath,amsfonts,amssymb,amsthm,amscd}

\usepackage[T1]{fontenc}
\usepackage{ae} 
\usepackage{aecompl} 

\usepackage{epsfig}
\usepackage{psfrag}

\usepackage{hyperref}
\usepackage{breakurl}

\newtheorem{theorem}{Theorem}
\newtheorem{lemma}[theorem]{Lemma}
\newtheorem{proposition}[theorem]{Proposition}
\newtheorem{corollary}[theorem]{Corollary}

\newtheorem{observation}[theorem]{Observation}
\newtheorem{claim}[theorem]{Claim}
\newtheorem{problem}{Problem}

\def\inst#1{$^{#1}$}

\begin{document}


\title{Graph sharing games: complexity and connectivity~\thanks{The research
was supported by the project CE-ITI (GACR P2020/12/G061) of the Czech Science Foundation
and by the grant SVV-2010-261313 (Discrete Methods and Algorithms). Viola M\'{e}sz\'{a}ros was also partially supported by OTKA Grant K76099 and by the grant no. MSM0021620838 of the Ministry of Education of the Czech Republic. Josef Cibulka and Rudolf Stola\v{r} were also supported by the Czech Science Foundation under the contract no.\ 201/09/H057. Rudolf Stola\v{r} was also supported by the Grant Agency of the Charles University, GAUK project number 66010.
} 
} 

\author{Josef Cibulka\inst{1}, Jan Kyn\v{c}l\inst{2}, Viola M\'{e}sz\'{a}ros\inst{2,3}, \\
Rudolf Stola\v{r}\inst{1} and Pavel Valtr\inst{2}
} 

\date{}

\maketitle

\begin{center}
{\footnotesize
\inst{1} 
Department of Applied Mathematics, \\
Charles University, Faculty of Mathematics and Physics, \\
Malostransk\'e n\'am.~25, 118~ 00 Praha 1, Czech Republic; \\ 
\texttt{cibulka@kam.mff.cuni.cz, ruda@kam.mff.cuni.cz} 
\\\ \\
\inst{2}
Department of Applied Mathematics and Institute for Theoretical Computer Science, \\
Charles University, Faculty of Mathematics and Physics, \\
Malostransk\'e n\'am.~25, 118~ 00 Praha 1, Czech Republic; \\
\texttt{kyncl@kam.mff.cuni.cz}
\\\ \\
\inst{3}
Bolyai Institute, University of Szeged, \\
Aradi v\'ertan\'uk tere 1, 6720 Szeged, Hungary; \\
\texttt{viola@math.u-szeged.hu}
}
\end{center}  

\begin{abstract}
We study the following combinatorial game played by two players, Alice and Bob, which generalizes the Pizza game considered by Brown, Winkler and others. Given a connected graph $G$ with nonnegative weights assigned to its vertices, the players alternately take one vertex of $G$ in each turn. The first turn is Alice's. The vertices are to be taken according to one (or both) of the following two rules: (T) the subgraph of $G$ induced by the taken vertices is connected during the whole game, (R) the subgraph of $G$ induced by the remaining vertices is connected during the whole game. We show that if rules (T) and/or (R) are required then for every $\varepsilon > 0$ and for every $k\ge 1$ there is a $k$-connected graph $G$ for which Bob has a strategy to obtain $(1-\varepsilon)$ of the total weight of the vertices. This contrasts with the original Pizza game played on a cycle, where Alice is known to have a strategy to obtain $4/9$ of the total weight.

We show that the problem of deciding whether Alice has a winning strategy (i.e., a strategy to obtain more than half of the total weight) is PSPACE-complete if condition (R) or both conditions (T) and (R) are required. We also consider a game played on connected graphs (without weights) where the first player who violates condition (T) or (R) loses the game. We show that deciding who has the winning strategy is PSPACE-complete.
\end{abstract}

\section{Introduction}

\subsubsection*{The pizza problem.}
Dan Brown devised the following pizza puzzle in 1996 which we formulate using graph notation. The pizza with $n$ slices of not necessarily equal size, can be considered as a cycle $C_n$ with nonnegative weights on the vertices. Two players, Bob and Alice, are sharing it by taking turns alternately. In every turn one vertex is taken. The first turn is Alice's. Afterwards, a player can take a vertex only if the subgraph induced by the remaining vertices is connected. Dan Brown asked if Alice can always obtain at least half of the sum of the weights. Peter Winkler and others solved this puzzle by constructing a weighted cycle where Bob had a strategy to get at least $5/9$ of the total weight~\cite{czechpizza}. Peter Winkler conjectured that Alice can always gain at least $4/9$ of the total weight. This conjecture has been proved independently by two groups of authors~\cite{czechpizza,germanpizza}.



\subsubsection*{The games T, R and TR.} 
In this paper we investigate a generalized setting where the cycle $C_n$ is replaced with an arbitrary connected graph $G$. Consider the following two conditions:

\begin{enumerate}
\item[(T)] the subgraph of $G$ induced by the {\em taken\/} vertices is connected during the whole game, 

\item[(R)] the subgraph of $G$ induced by the {\em remaining\/} vertices is connected during the whole game. 
\end{enumerate}

Observe that the two conditions are equivalent if $G$ is a cycle or a clique.

The generalized game is called {\em game T}, {\em game R}, or {\em game TR\/} if we require condition (T), condition (R), or both conditions (T) and (R), respectively. Note that in games T and R, regardless of the strategies of the players, there is always an available vertex for the player on turn until all vertices are taken as the graph $G$ is connected. This need not be the case in game TR, however. Therefore we consider game TR played only on $2$-connected graphs, where the game never ends prematurely; see Lemma~\ref{lemma_2connected}. In Proposition~\ref{prop_TR} we give a complete characterization of the graphs $G$ where the game TR ends with all vertices taken for all strategies of the players. 

\subsubsection*{Playing on $k$-connected graphs.} 
In games T, R and TR, there are graphs where Bob can ensure (almost) the whole weight to himself. See Fig.~\ref{f:trees} for examples in games T and R where the graph $G$ is a tree. For any $k\ge 2$, examples of $k$-connected  graphs $G$ exist in any of the three variants of the game. See Fig.~\ref{f:2connected} for examples of such $2$-connected graphs.

\begin{figure}
\begin{center}
\includegraphics[scale=1]{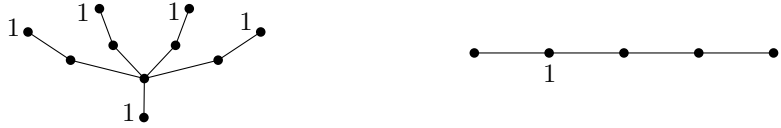}
\end{center}
\caption{Left, game T: a tree where Alice gets at most one vertex of weight $1$. Right, game R: a path of even length, Bob gets the only vertex of positive weight. Vertices with no label have weight $0$.}

\label{f:trees}
\end{figure}

\begin{figure}
\begin{center}
\includegraphics[scale=1]{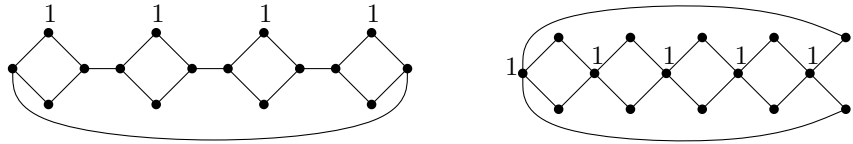}
\end{center}
\caption{Examples of $2$-connected graphs where Alice gets at most one vertex of weight $1$. Left: game T, or game TR. Right: game R; the number of vertices of degree $4$ (and of weight $1$) must be odd. Vertices with no label have weight $0$.}

\label{f:2connected}
\end{figure}

\begin{theorem}\label{veta_kconnected}
For each of the games T, R and TR, for every $\varepsilon > 0$ and for every $k\ge 1$ there is a $k$-connected graph $G$ for which Bob has a strategy to obtain at least $(1-\varepsilon)$ of the total weight of the vertices.
\end{theorem}

\subsubsection*{Parity of the vertex set.} 
Micek and Walczak~\cite{graph_grabbing,parity_sharing} also studied, independently of us, generalizations of the pizza game. They investigated how the parity of the number of vertices affects Alice's chances to gain a positive fraction of the total weight in games T and R, particularly when the game is played on a tree. They proved that Alice can gain at least $1/4$ of the total weight in game R on a tree with an even number of vertices and in game T on a tree with an odd number of vertices. For the opposite parities they constructed examples of trees (such as those in Figure~\ref{f:trees}) where Bob has a strategy to get almost all the weight. They conjectured that Alice can gain at least $1/2$ of the total weight in game R on a tree with an even number of vertices; this has been confirmed by Seacrest and Seacrest~\cite{seacrest}.

In our original proof of Theorem~\ref{veta_kconnected} the graphs for games T and TR had an even number of vertices and the graphs for game R an odd number of vertices.
Micek and Walczak (personal communication) asked whether also $k$-connected graphs with opposite parities satisfying the conditions of Theorem~\ref{veta_kconnected} exist. 
In Subsection~\ref{sub_2_2} we find examples of such graphs for all three variants of the game. 

Micek and Walczak~\cite{parity_sharing} also conjectured that there is a function $f(m) > 0$ such that in game T, Alice can gain at least $f(m)$ of the weight of any graph with an odd number of vertices and with no $K_m$-minor. G\k{a}gol~\cite{gagol} showed that under the additional assumption that the input graph $G$ has only $\{0,1\}$ weights, Alice gan gain at least 
$\Omega(1/(m\sqrt{\log m})^3)$ of the total weight.
G\k{a}gol, Micek and Walczak~\cite{GMW12_subdivision} recently informed us about proving a stronger version of their conjecture: there is a function $f(m)>0$ such that in game T, Alice can
gain at least $f(m)$ of the weight of any graph with an odd number of
vertices and with no subdivision of $K_m$. On the other hand, they show that the construction of odd graphs which are arbitrarily bad for Alice can be easily modified to yield graphs with bounded expansion.

Micek and Walczak~\cite{graph_grabbing} noted that all known examples of sequences of even graphs where Alice's gain in game R tends to zero, contain arbitrarily large cliques as subgraphs. We show that the condition of containing arbitrarily large cliques is not necessary to force Alice's gain to tend to zero. In fact, the graphs in our example have arbitrarily large girth and form a class with bounded local expansion. 

\begin{theorem}\label{veta_BE}
For every $\varepsilon>0$ there is a $\{0,1\}$-weighted even graph $G$ with arbitrarily large girth for which in game R, Bob has a strategy to obtain at least $(1-4\varepsilon)$ of the total weight of the vertices. Moreover, 
\begin{enumerate}

\item[{\rm(1)}] there exists an infinite class $\mathcal{G}_{\varepsilon}$ of such graphs with maximum degree bounded by $O(\frac{1}{\varepsilon}\log(\frac{1}{\varepsilon}))$; in particular, $\mathcal{G}_{\varepsilon}$ has bounded expansion. 

\item[{\rm(2)}] There is a class $\mathcal{G}$ with bounded local expansion of $\{0,1\}$-weighted even graphs for which the infimum of the fraction of Alice's gain in game R is zero.
\end{enumerate}
\end{theorem}

For definitions of the classes with {\em bounded expansion\/} and {\em bounded local expansion}, see~\cite{BLE}.

\subsubsection*{Canonical game and misere game.} 
As game TR does not always end with all vertices taken, it is natural to consider the following variation of the game, which we call a {\em canonical game TR}. Given a connected graph $G$, Alice and Bob take turns alternately. In each turn a player takes one vertex of $G$. The first player who has no move satisfying conditions (T) and (R) loses the game. We also consider a {\em misere game TR}, where the first player who has no move satisfying conditions (T) and (R) wins the game. 

\subsubsection*{Complexity results and open problems.}
We determine the complexity of deciding who has the winning strategy in the canonical and the misere game.

\begin{theorem}\label{veta_canonical_pspace}
It is PSPACE-complete to decide who has the winning strategy in the canonical game TR and in the misere game TR.
\end{theorem}

We also consider the complexity of determining the winning strategy (i.e., gaining more than half of the total weight) for the original three types of the game. We show the following.

\begin{theorem}\label{veta_weighted_pspace}
It is PSPACE-complete to decide who has the winning strategy in game R and TR.
\end{theorem}

However, we are unable to determine the complexity of deciding the winner for game T.

\begin{problem}
What is the complexity of deciding who has the winning strategy in game T?
\end{problem}

\begin{problem}
What is the complexity of deciding who has the winning strategy in game R and in game T when the input graph $G$ is a tree?
\end{problem}

We are not able to show that the problem is polynomial even if the tree has only one branching vertex. 

Game $R$ played on a tree is also known as the {\em gold-grabbing game\/}~\cite{gold_grabbing, seacrest}.

Since the weights used in our proof of Theorem~\ref{veta_weighted_pspace} are growing exponentially, it is natural to ask the following question.

\begin{problem}
What is the complexity of deciding who has the winning strategy in game R and TR, when the weights of the vertices of $G$ are only $0$ or $1$?
\end{problem}


\section{Proofs}

\subsection{Graphs where game TR ends with all vertices taken}

Here we characterize the graphs where game TR always ends with all vertices taken, and also the graphs for which there is at least one sequence of turns in game TR such that the game ends with all vertices taken. First we prove an auxiliary observation implying that it is always ``safe'' to play game TR on a $2$-connected graph.

\begin{lemma}\label{lemma_2connected}
Let $G$ be a $2$-connected graph with $n$ vertices and let $v$ be a vertex of $G$. Let $C$ be a set of $i\le n-2$ vertices of $G$ such that $v\notin C$ and the two induced subgraphs $H_i=G[C]$ and $H'_i=G[V(G)\setminus C]$ are connected. Then there is a vertex $v_{i+1} \in V(G)\setminus (C \cup \{v\})$ such that both induced subgraphs $H_{i+1}=G[C\cup \{v_{i+1}\}]$ and $H'_{i+1}=G[V(G)\setminus (C \cup \{v_{i+1}\})]$ are connected.

\end{lemma}

\begin{proof}
The $2$-connectivity of $G$ implies that $H'_i$ has at least two vertices that are neighbors of $H_i$. Let $w_1,w_2$ be such a pair of vertices with the largest possible distance in $H'_i$. We claim that neither of the vertices $w_1,w_2$ separates $H'_i$ and therefore we can choose $v_{i+1} \in \{w_1,w_2\} \setminus \{v\}$. Suppose for contradiction that $w_1$ is a cut vertex of $H'_i$ and let $C_1$ be a component of $H'_i-w_1$ that does not contain $w_2$. Since $w_1$ does not separate $G$, the component $C_1$ contains a neighbor $w_3$ of $H_i$. The shortest path in $H'_i$ between $w_2$ and $w_3$ passes through $w_1$, which contradicts the choice of the pair $w_1,w_2$.
\end{proof}

\begin{corollary}\label{cor_order}
Let $G$ be a $2$-connected graph with $n$ vertices and let $u$ and $v$ be two distinct vertices of $G$. Then the vertices of $G$ can be ordered as $u_1=u,u_2,\dots,u_n=v$ so that for each $i=1,2,\dots,{n-1}$, both induced graphs $G[\{u_1,u_2,\dots,u_i\}]$ and $G[\{u_{i+1},\dots,u_{n-1},u_n\}]$ are connected.
\qed
\end{corollary}

\begin{proposition}\label{prop_TR}
{\rm (1)} For game TR on a graph $G$ there is a sequence of turns to take all the vertices of $G$ if and only if each cut vertex of $G$ separates the graph into precisely two components and every $2$-connected component of $G$ contains at most two cut vertices of $G$. 

{\rm (2)} Game TR on a graph $G$ will always end with all vertices taken if and only if each cut vertex of $G$ separates the graph into precisely two components and every $2$-connected component of $G$ with at least three vertices contains at most one cut vertex of $G$.
\end{proposition}

\begin{proof}
\noindent (1) First we show that the two conditions are necessary. 

Suppose $v$ is a cut vertex separating $G$ into more than two components. Let $C$ be the component where Alice made her first turn, and let $C_1$ and $C_2$ be two other components. By rule (T), before $v$ is taken, only vertices of $C$ can be taken. Therefore by rule (R), the vertex $v$ never becomes available since it separates $C_1$ from $C_2$ in the remaining graph.

Let $C$ be a $2$-connected component of $G$. Suppose for contradiction that $C$ contains at least three cut vertices of $G$. Then $G-C$ has at least three components $C_1,C_2,C_3$. Let $v_i$ be the only vertex of $C$ neighboring with $C_i$ ($i=1,2,3$). By rule (T), before any vertex of $C$ is taken, the players can take vertices from at most one of the components $C_1,C_2,C_3$, say $C_1$. Before any vertex of $C_2\cup C_3$ is taken, one of the vertices $v_2,v_3$ has to be taken. But these vertices never become available by rule (R), since they both separate $C_2$ from $C_3$.

If both conditions are satisfied, then the $2$-connected components of $G$ can be arranged into a sequence $C_1, C_2, \dots, C_k$, where for each $i=1,2,\dots k-1$, the components $C_i$ and $C_{i+1}$ share a cut vertex $v_i$. 
By applying Corollary~\ref{cor_order} for each of the components $C_i$ and the cut vertices it contains, we obtain the following order in which the players can take the vertices:
$$u_{1,1},u_{1,2},\dots,u_{1,n_1},v_1,u_{2,1},u_{2,2},\dots,u_{2,n_2},v_2,\dots,v_{k-1},u_{k,1},u_{k,2},\dots,u_{k,n_k},$$ 
where 
\begin{align*}
\{u_{1,1},u_{1,2},\dots,u_{1,n_1},v_1\}&=C_1, \\
\{v_{i-1},u_{i,1},u_{i,2},\dots,u_{i,n_i},v_i\} &= C_i,\ {\rm for}\ i=2,\dots,k-1,\ {\rm and} \\
\{v_{k-1},u_{k,1},u_{k,2},\dots,u_{k,n_k}\} &= C_k.
\end{align*}

\noindent (2) If game TR always ends with all vertices taken, then we may assume that the two conditions from part (1) are satisfied. Suppose that $G$ has a $2$-connected component $C$ with at least $3$ vertices and with two cut vertices of $G$: a vertex $v_1$ separating $C$ from $C_1$ and $v_2$ separating $C$ from $C_2$. If Alice starts with taking a vertex from $C\setminus\{v_1,v_2\}$, then neither of the vertices $v_1, v_2$ becomes available during game TR.

If both conditions are satisfied, then the graph $G$ is $2$-connected or it is a union of two $2$-connected subgraphs $H, H'$ (including the degenerate cases when $H$ or $H'$ has only one vertex) and a path $P=v_1,v_2, \dots,v_k$ of cut vertices where $P\cap H=\{v_1\}$ and $P\cap H'=\{v_k\}$. By rule (R), Alice has to start in $H-v_1$ or $H'-v_k$ (or in $v_1$ or $v_k$ in the degenerate cases). Suppose without loss of generality that she starts with taking a vertex $u$ in $H-v_1$. By rule (R), all vertices of $H-v_1$ have to be taken before $v_1$. By Lemma~\ref{lemma_2connected} applied to the set $V(H)$ and the vertex $v_1$, all the vertices of $H$ will be indeed taken. After that the players have to take the vertices of the path $P$ sequentially from $v_1$ to $v_k$. If $H'$ has at least three vertices, then by Lemma~\ref{lemma_2connected}, also all the vertices of $H'$ will be taken before the game ends.
\end{proof}


\subsection{Proof of Theorem~\ref{veta_kconnected}}\label{sub_2_2}

In games T and TR, Bob can choose the following $k$-connected graph with an even number of vertices (for any given $k\ge 2$):
Take a large even cycle and replace each vertex in it by a $2\lceil k/4\rceil$-clique and each edge by a complete bipartite graph $K_{2\lceil k/4\rceil,2\lceil k/4\rceil}$.
Assign weight $1$ to one vertex in every other $2\lceil k/4\rceil$-clique, and weight $0$
to all the other vertices of the graph. See Figure~\ref{obr_kconnected}, left. Bob uses the following strategy.
\begin{enumerate}
\item[1)] Take an availabe vertex of weight $1$. 
\item[2)] If no vertex of weight $1$ is available, take an available vertex from one of the $2\lceil  k/4\rceil$-cliques where at least one vertex has already been taken. 
\end{enumerate}
It is easy to see that by this strategy Bob makes sure that Alice takes at most one vertex of weight $1$, and only in her first or second turn.

In game R, Bob can choose the following $k$-connected (bipartite) graph $G$ with an odd number of vertices. The vertex set is a disjoint union of sets $X,Y$ and $Z$ where $Y$ is an $m$-element set for some large $m\ge k+2$, $Z$ is the set of all $k$-element subsets of $Y$ and $X$ is a set of $|Y|+|Z|+2$ or $|Y|+|Z|+3$ elements so that the total number of vertices is odd. The edge set of $G$ consists of all edges between $X$ and $Y$ and all edges that connect a vertex $z \in Z$ with each of the $k$ vertices $y \in Y$ such that $y\in z$. Each vertex from $Y$ has weight $1$, all the other vertices have weight $0$. See Figure~\ref{obr_kconnected}, right.

\begin{figure}
\begin{center}
\includegraphics[scale=1]{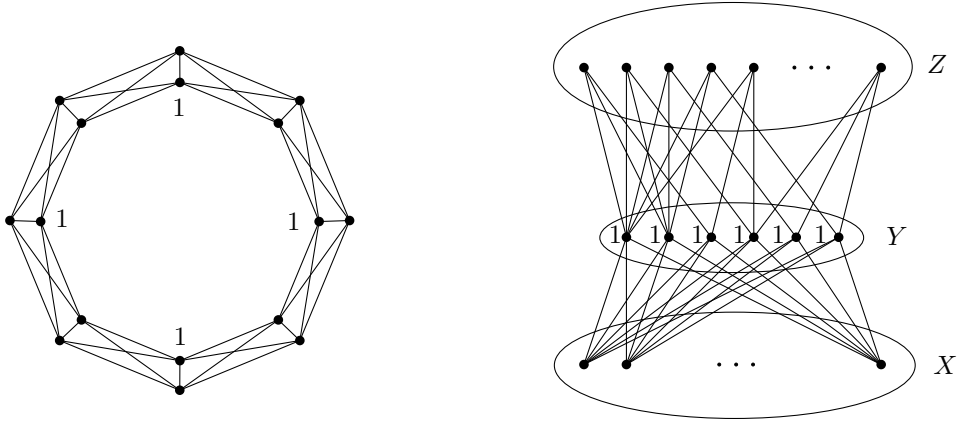}
\end{center}
\caption{Left: a $4$-connected graph with an even number of vertices where Alice gets at most one vertex of weight $1$ in games T and TR. Right: a $3$-connected graph with an odd number of vertices where Alice gets at most one vertex of weight $1$ in game R; Vertices with no label have weight $0$.}

\label{obr_kconnected}
\end{figure}

Consider a position in the game. Let $H$ be a subgraph of $G$ induced by the remaining vertices. A vertex $v \in H$ is {\em available\/} if it is not a cut vertex of $H$ (equivalently, $v$ can be taken in the next turn). A {\em leaf\/} is a vertex in $H$ of degree one. 

\begin{claim}
Bob can force Alice to get at most $\lfloor k/2\rfloor$ vertices of weight $1$.
\end{claim}

\begin{proof}
Bob has the following strategy.
\begin{enumerate}
\item[1)] If possible, he takes an available vertex of $Y$. 
\item[2)] If no vertex of $Y$ is available, then if possible, he takes either a vertex of $Z$ that is not a leaf or a leaf vertex of $Z$ whose neighbor in $Y$ neighbors at least one other leaf in $Z$.
\item[3)] If neither 1) nor 2) applies, he takes any vertex of $X$.
\end{enumerate}


There are three phases of the game. The $i$th phase lasts as long as Bob acts only according to the rules with numbers at most $i$.

During the first phase Bob takes only vertices from $Y$.
Let $a$ and $b$ be the numbers of vertices from $Y$ taken in the first phase
by Alice and by Bob, respectively. Since the first phase ends by Alice's
turn, Alice takes exactly $b+1-a$ vertices of $X\cup Z$ in the first
phase. Due to rules (R) and 1), all the vertices of $Z$ whose all $k$
neighbors are taken in the first phase must be taken by Alice in the
first phase. There are ${a+b\choose k}$ such vertices.
It follows that ${a+b\choose k}\le b-a+1$. Equivalently,
$$2a\le a+b-{a+b\choose k}+1.$$
If $a+b\in\{k-1,k\}$ then the right-hand side equals $k$; otherwise it is
smaller than $k$. Thus, $2a\le k$. It means that Alice takes at most
$\lfloor k/2\rfloor$ vertices of weight $1$ in the first phase.
We further show that Alice takes only vertices of weight $0$ in the
other two phases.


At the beginning of the second phase each remaining $y \in Y$ neighbors at least one leaf in $Z$. During the second phase this property is preserved because in each turn at most one vertex from $Y$ may become available. Note that only Alice can make a vertex $y$ in $Y$ available, so Bob will take such $y$ in the following turn. Thus Alice is forced to take vertices only from $X \cup Z$ during the second phase. 

As $X$ has at least two more vertices than $Y\cup Z$, at some point the third phase has to start. At the beginning of the third phase the remaining vertex set is a union of three sets $X',Y'$ and $Z'$ that are subsets of $X,Y$ and $Z$, respectively. There is a complete bipartite graph on $(X', Y')$ and there is a matching between $Y'$ and $Z'$. The first turn in the third phase is Bob's. He takes a vertex from $X'$ (there is an available vertex in $X'$ since the size of $X'$ is positive and even). Alice may take vertices from $X' \cup Z'$. Whenever Alice takes a vertex of $X'$, then in the consecutive turn Bob will do the same. Whenever Alice takes a vertex from $Z'$, then in the consecutive turn Bob will take the available vertex from $Y'$. This implies that either the entire $Y' \cup Z'$ is taken in the previously described way or the graph $G$ transforms into a star centered in the only remaining vertex of $X'$ with each edge subdivided by a vertex in $Y'$, and with each leaf in $Z'$. As the original graph $G$ had odd number of vertices, the next turn is Alice's and consequently Bob collects all the remaining vertices of $Y$. The claim follows.
\end{proof}

Now we show the constructions of graphs with the parities opposite to those in the proof above.

For game T and game TR and for every $k\ge 1$, we can construct a $(2k+1)$-connected graph $H_{n,k}$ with an odd number of vertices starting from the graph $H_n$ described by Micek and Walczak~\cite[Example 2.2]{parity_sharing}, replacing each vertex of weight $0$ by $2k+1$ vertices of weight $0$ forming a clique, and replacing each original edge by a complete bipartite graph. 
The graph $H_n$ consists of vertices $a_1, a_2, \dots, a_n$ of weight $1$, vertices $b_1, b_2, \dots, b_n$ of weight $0$, and a vertex $c_S$ of weight $0$ for every non-empty subset $S\subseteq \{1,2,\dots, n\}$. Each $a_i$ is joined by an edge to $b_i$ and each $c_S$ is joined to all $b_i$ such that $i\in S$. For the graph $H_{n,k}$ Bob has a strategy to take all but one vertex of weight $1$, analogous to the strategy for $H_n$~\cite{parity_sharing}.

For game R, for every $k\ge 1$ we construct a $k$-connected weighted graph $G'_{n,k}$ with $n + {n\choose k}$ vertices (we may assume that $n=2^m>k$ for some positive integer $m$ so that the total number of vertices is even).
The construction generalizes the graph $G'_n$~\cite[Example 5.2]{parity_sharing} consisting of a clique of $n$ vertices of weight $1$, with a leaf of weight $0$ attached to each vertex of the clique. 
The graph $G'_{n,k}$ consists of a clique on $n$ vertices $a_1, a_2, \dots, a_n$ of weight $1$ and a vertex $b_S$ of weight $0$ for each $k$-element subset $S \subseteq \{1,2, \dots, n\}$. The vertex $b_S$ is connected by an edge to all $k$ vertices $a_i$ such that $i \in S$.

Alice can collect at most $\lfloor k/2 \rfloor + 1$ vertices of weight $1$ if Bob plays as follows. 
\begin{enumerate}
\item[1)] If possible, Bob takes a vertex of weight $1$. 
\item[2)] Otherwise he takes a vertex of weight $0$ which is not a unique leaf neighbor of some vertex $a_i$ (in the graph induced by the remaining vertices). 
\end{enumerate}
Bob can always play according to one of these two rules because the number of remaining vertices before his turn is odd. It follows that Bob plays in such a way that no vertex of weight $1$ becomes available after his turn, except the turn after which only two vertices remain.
By the same argument as in the proof of Theorem~\ref{veta_kconnected}, Alice takes at most $\lfloor k/2 \rfloor$ vertices of weight $1$ at the beginning of the game, while Bob plays only by rule 1). Then she takes at most one more vertex of weight $1$ in her last turn.


\subsection{Proof of Theorem~\ref{veta_BE}}\label{sub_2_3}

Let $\varepsilon>0$ be fixed. We construct a weighted even graph $G$ as follows. The graph $G$ consists of a graph $H$ with all vertices of weight 1, and one leaf of weight 0 connected to each vertex of $H$. In a similar construction of the graph $G'_n$ by Micek and Walczak~\cite[Example 5.2]{parity_sharing}, $H$ was a complete graph. Here we take as $H$ a much sparser graph on $n$ vertices, which is still a ``good expander'' in the following sense: the complement of $H$ contains no complete bipartite subgraph $K_{\lfloor\varepsilon n\rfloor, \lfloor\varepsilon n\rfloor}$. In addition, $H$ is a graph with girth $\Omega(\log n)$ and we may also require $H$ to have a bounded maximum degree.
We show that such a graph $H$ exists using probabilistic method.

In the rest of this section we omit the explicit rounding in the expressions involving $\varepsilon n$, to keep the notation simple.

\begin{lemma}\label{lemma_nahodne}
Let $\varepsilon\in (0,1)$, $c=2\log(3e/\varepsilon)/\varepsilon$ and $c'= 1/(2\log(3c))$. There exists $n_0(\varepsilon)$ such that for every $n>n_0(\varepsilon)$ there exists a connected graph $H$ with $n$ vertices of girth at least $c'\log{n}$ and of maximum degree at most $9c$, whose complement contains no copy of $K_{\varepsilon n,\varepsilon n}$. 
\end{lemma}

\begin{proof}
Let $H'$ be a random graph from the probability space $G(3n, c/n)$. That is, $H'$ is a subgraph of  $K_{3n}$ where every edge is taken independently with probability $c/n$. First we show that the  expected number of cycles of length at most $c'\log n$ in $H'$ is less than $n$. The number of cycles  of length $k$ in the complete graph with $3n$ vertices is ${3n \choose k} k! / (2k) \le (3n)^k$. A  cycle of length $k$ appears in $H'$ with probability $(c/n)^k$. Hence, the expected number of cycles of  length $k$ is at most $(c/n)^k (3n)^k = (3c)^k$. Therefore, the expected number of cycles of length at  most $c'\log n$ in $H'$ is at most
$$\sum_{k=3}^{c'\log n} (3c)^k \le (3c)^{c'\log n+1} \le 3cn^{c'\log (3c)} = 3cn^{1/2} = o(n).$$

The expected number of subgraphs $K_{\varepsilon n,\varepsilon n}$ in the complement of $H'$ is at most 
$$\left(1-\frac{c}{n}\right)^{(\varepsilon n)^2} {3n \choose 2\varepsilon n}{2\varepsilon n\choose  \varepsilon n} \le 
e^{-c\varepsilon^2 n} (3n)^{2\varepsilon n}(e/\varepsilon n)^{2\varepsilon n}$$
$$\le e^{-c\varepsilon^2 n} (3e/\varepsilon)^{2\varepsilon n} \le e^{\varepsilon  n(-c\varepsilon+2\log(3e/\varepsilon))} = 1.$$

The expected average degree of $H'$ is smaller than $3c$ but the expected maximum degree of $H'$ is  unbounded. However, the number of vertices of large degree is small. The probability that a given  vertex has degree larger than $9c$ is at most 
$${3n \choose 9c}\left(\frac{c}{n}\right)^{9c}\le  (3n)^{9c}\left(\frac{e}{9c}\right)^{9c}\left(\frac{c}{n}\right)^{9c}\le\left(\frac{e}{3}\right)^{9c}<\frac{1}{3}.$$

(This is also a direct consequence of Markov's inequality).
The expected number of vertices of degree larger than $9c$ is thus at most $n$. 

It follows that there exists a graph with $3n$ vertices such that by deleting some $2n$ vertices, we  obtain a graph $H$ with $n$ vertices, with maximum degree at most $9c$, with no cycle shorter than  $c'\log n$, and with no $K_{\varepsilon n, \varepsilon n}$ in the complement. 

In case the graph $H$ is not connected, we add the necessary edges connecting different components of  $H$ to make $H$ connected, in such a way that the maximum degree of $H$ does not exceed $9c+1$.
\end{proof}

Let $\varepsilon>0$ be fixed. By taking a graph $H$ from Lemma~\ref{lemma_nahodne} for every $n>n_0(\varepsilon)$ and attaching a leaf to each vertex of $H$, we get an infinite class $\mathcal{G}_{\varepsilon}$ of graphs with bounded maximum degree, thus with bounded expansion. Note that to obtain a class of bounded expansion it is not necessary to delete vertices of high degree as in the proof of Lemma~\ref{lemma_nahodne}, since a.a.s. the random graphs $G(n,c/n)$ form a class with bounded expansion~\cite{BE}.

The class $\mathcal{G}$ is constructed as follows. For every $\varepsilon>0$, let $c,c'$ and $n_0(\varepsilon)$ be as in Lemma~\ref{lemma_nahodne}. Let $n>n_0(\varepsilon)$ be such that $c'\log n > 9c+1$ and let $H=H_{\varepsilon,n}$ be the graph given by Lemma~\ref{lemma_nahodne}. Let $G_{\varepsilon,n}$ be the graph obtained by attaching a leaf to each vertex of $H_{\varepsilon,n}$. Then $\mathcal{G}=\{G_{1/m,n}, m=1,2,\dots\}$ is a class of graphs satisfying $\mathrm{girth}(G)\ge \Delta(G)$ and therefore has bounded local expansion~\cite{BLE}.

Now we show that Alice's gain in game R played on the graph $G$ is bounded by $4\varepsilon$ of the total weight. That is, Alice takes at most $4\varepsilon n$ vertices of $H$ during the game. We need only the property that the complement of $H$ contains no $K_{\varepsilon n, \varepsilon n}$ and that $H$ is connected.

During the game, we call a remaining vertex $v$ of $H$ {\em exposed\/} if its neighboring leaf of weight $0$ has been taken. Let $H_R$ be the remaining subgraph of $H$. Let $B$ be the set of exposed cut vertices of $H_R$.  
Let $K = H_R - B$. Call a component of the graph $K$ {\em large\/} if it has at least $\varepsilon n$ vertices and {\em small\/} otherwise. Observe that $K$ has at most one large component, which we denote by $L$.
Let $B_L \subseteq B$ be the set of exposed cut vertices adjacent to $L$. Let $S=H_R-L-B_L$. Observe that $|B_L| \le |V(S)| < \varepsilon n$.

Bob's strategy is the following.

\begin{enumerate}
\item[1)] If Alice took a vertex of weight 0 neighboring a vertex $v \in L$
in the previous turn, Bob takes $v$, if it is available (that is, if $v$ is not a cut vertex in $H_R$). 

\item[2)] If Alice makes a vertex $v\in B_L$ available (in which case $v$ does not belong to $B$ after her turn), Bob takes $v$.

\item[3)] If neither of the previous rules apply, Bob takes any vertex $v$ such that $v \notin L$, $v$ is not a vertex of weight $0$ neighboring a vertex of $L$, and taking $v$ does not make any vertex of $B_L$ available. In case there is no large component $L$, Bob may take any available vertex.
\end{enumerate}

We claim that Bob's strategy is complete, that is, he can always play by one of the rules. We show this in a series of observations.

\begin{observation}\label{obs_L}
The large component contains at most one exposed vertex, and that can happen only after Alice's turn.
\end{observation}

\begin{proof}
This is a direct consequence of Bob's strategy: he is not allowed to expose a vertex of $L$, and by rule 1) he immediately takes a vertex of $L$ exposed by Alice.
\end{proof}

\begin{corollary}\label{cor_A_not_L}
Alice never takes a vertex of $L$.
\qed
\end{corollary}

\begin{observation}\label{obs_B_L}
A vertex of $B_L$ can become available only after Alice's turn and Bob takes it immediately in his next turn.
\end{observation}

\begin{proof}
By following rules 1) and 2) Bob takes a vertex from $L$ or $B_L$ and thus cannot make a vertex of $B_L$ available. Rule 3)  explicitly forbids making a vertex of $B_L$ available. Alice can make a vertex of $B_L$ available only by taking its neighbor from $H$. For Bob's next turn only rule 2) applies and the observation follows.
\end{proof}

\begin{corollary}\label{cor_A_not_B_L}
Alice never takes a vertex of $B_L$.
\qed
\end{corollary}

\begin{observation}
If Bob has to follow rule 3), there is a vertex $v$ he can take.
\end{observation}

\begin{proof}
Suppose that Bob is forced to follow rule 3) and that $H_R$ has a large component $L$. By Observation~\ref{obs_L}, none of the vertices of $L$ is exposed. In particular, the total number of vertices in $L$ and the remaining leaves of weight $0$ attached to $L$ is even. If taking every available vertex $v$ makes some vertex $u\in B_L$ available, then $v$ is the only neighbor of $u$ outside $L$. In particular, $|B|=|B_L|=|V(S)|$ and hence the total number of vertices of $H_R$ is even, which is a contradiction with Bob being on turn. Therefore there is an available vertex $v$ satisfying the conditions of rule 3).
\end{proof}


It remains to estimate Alice's gain. To this end, we need an upper bound on the size of the ``latest'' large component. Let $L_i$, $i=0,1,2,\dots, f$, be the large component after Bob's $i$th turn. The number $f$ is chosen as the largest possible.

\begin{observation}
For $i=1,2,\dots, f$, we have $L_i\subseteq L_{i-1}$. There is no large component after $f+1$ or more Bob's turns. 
\end{observation}

\begin{proof}
The large component expands only when some vertex $v$ of $B_L$ becomes available (and thus $v$ is added to $L$). By Observation~\ref{obs_B_L}, $v$ is taken by Bob in the next turn and thus the large component remains the same as two turns before. Similar observation applies also for small components.
\end{proof}

The large component shrinks after Alice's and Bob's $i$th turns (that is, $L_i\subset L_{i-1}$) if either some of its vertices is taken or if some of its cut vertices becomes exposed. In the first case the size of $L$ drops exactly by $1$, in the second case the size of $L$ can drop by up to $\varepsilon n$. 

Now consider the last large component $L_f$. If some of the vertices of $L_f$ is taken in the next Bob's turn, then $L$ has precisely $\varepsilon n$ vertices, since the remaining subset of $L_f$ after Bob's turn is a small component of $H_R$. 

Suppose that a cut vertex $v$ of $L_f$ is exposed in the next Alice's turn. Then $L_f-v$ consists of small components only.

For $i=0,1,2,\dots, f$, let $S_i$ denote the graph $S$ after Bob's $i$th turn. Let $T=\bigcup_{i=0}^f S_i$. By Corollaries~\ref{cor_A_not_L} and~\ref{cor_A_not_B_L}, the only vertices of weight $1$ Alice can take belong to $T\cup L_f\cup B_{L_f}$. The following lemma thus provides an upper bound on Alice's gain.

\begin{lemma}\label{lemma_eps_3eps}
$|T|<\varepsilon n$ and $|L_f-v|+|T|<3\varepsilon n$.
\end{lemma}

\begin{proof}
Observe that there is no edge between a vertex of $T$ and a vertex of $L_f$. Indeed, every vertex of $T$ belongs to $S_i$ for some $i\le f$, $L_f \subseteq L_i$ and $L_i$ is non-adjacent to $S_i$ by the definition of $S_i$. Since $|V(L_f)|\ge \varepsilon n$, the inequality $|T|<\varepsilon n$ follows, otherwise the complement of $H$ would contain a forbidden copy of $K_{\varepsilon n,\varepsilon n}$.

The graph $K'=(L_f-v) \cup T$ consists of small components only. If $K'$ has at least $\varepsilon n$ vertices, consider a minimal subgraph $L'$ of $K'$ that has at least $\varepsilon n$ vertices and is a union of some components of $K'$. Then $L'$ has less than $2\varepsilon n$ vertices and $K'-L'$ has less than $\varepsilon n$ vertices, due to the forbidden $K_{\varepsilon n,\varepsilon n}$ in the complement of $H$. Therefore $K'$ has less than $3\varepsilon n$ vertices.
\end{proof}

By previous observations and Lemma~\ref{lemma_eps_3eps}, Alice's gain is at most 
$$|T|+|L_f|+|B_{L_f}|\le |T|+|L_f|+|S_f| \le |T|+|L_f|+|T| \le 4\varepsilon n.$$
This completes the proof of Theorem~\ref{veta_BE}.


\subsection{Proof of Theorem~\ref{veta_canonical_pspace}}
First we consider the canonical game TR.
We proceed by polynomial reduction from the standard PSPACE-complete problem TQBF (also called QBF).
An instance of the TQBF problem is a fully quantified boolean formula with $n$ variables and alternating quantifiers, and the question is whether $\Phi$ is true. Without loss of generality we may assume that $n$ is even and that the formula starts with the existential quantifier: 
$$\Phi = \exists x_1 \forall x_2 \exists x_3 \forall x_4 \dots \forall x_n \varphi(x_1, x_2, \dots, x_n).$$

We may also assume without loss of generality that $\varphi$ in the previous expression is a 3-SAT formula. 

As the game always ends after polynomially many turns, one can search through all possible game states and determine who has the winning strategy in PSPACE.
To show that the problem is also PSPACE-hard, it suffices to prove that for every formula $\Phi$ there is a graph $G_{\Phi}$ constructible in polynomial time such that $\Phi$ is true if and only if Alice has a strategy to win on $G_{\Phi}$ in the canonical game TR.


\subsubsection{The construction of $G_{\Phi}$}

First we introduce a {\em V-gadget\/} that will be used many times in the construction of $G_{\Phi}$. The V-gadget is a path $P$ of length four. The middle vertex $c$ of $P$ is distinguished because $c$ will be identified with other vertices during the construction. 

For every variable of $\varphi$ we build a {\em variable gadget}. For the variable $x_i$ the gadget consists of a path $P_i$ of length two between the vertices $T_i$ and $F_i$ that represent the two possible values of $x_i$, and we attach a V-gadget to the middle vertex of $P_i$, see Figure~\ref{f:var}. The variable gadgets are connected by edges $T_iT_{i+1}$, $T_iF_{i+1}$, $F_iT_{i+1}$ and $F_iF_{i+1}$ in $G_{\Phi}$ for all $i<n$, see Figure~\ref{f:clause}.

\begin{figure}
\begin{center}
\includegraphics[scale=1]{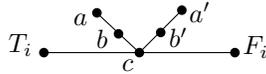}
\end{center}
\caption{Variable gadget for $x_i$ where $T_i$ and $F_i$ represent the two possible values, TRUE and FALSE, of the variable $x_i$.}
\label{f:var}
\end{figure}

For every clause $c_l$ of $\varphi$ we introduce a new vertex $C_l$ in $G_{\Phi}$. The vertex $C_l$ is connected to the three vertices in $G_{\Phi}$ corresponding to the literals of the clause $c_l$. In case a variable $x_i$ stands with negation in $c_l$, the vertex $C_l$ is connected to $T_i$ in $G_{\Phi}$, otherwise $C_l$ is connected to $F_i$. We attach a V-gadget to $C_l$. Further we add one special vertex $L$ to $G_{\Phi}$ and edges $T_nL$, $F_nL$, and $LC_l$ for each $C_l$, see Figure~\ref{f:clause}. 

\begin{figure}
\begin{center}
\includegraphics[scale=1]{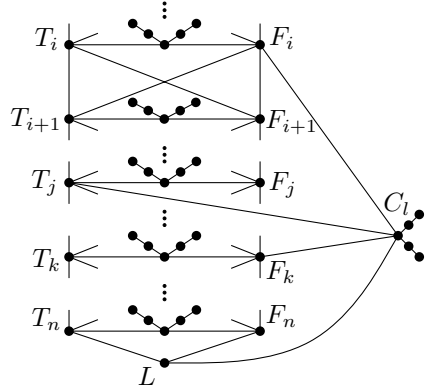}
\end{center}
\caption{Variable gadgets and the vertex $C_l$ corresponding to the clause $c_l=(x_i \vee \neg x_j \vee x_k)$.}
\label{f:clause}
\end{figure}

The {\em order enforcer\/} for a vertex $u$ is a gadget that prevents Alice from starting at $u$. This property is proved in Observation~\ref{obs_enforcer} below. The {\em special neighbors\/} of $u$ are the neighbors of $u$ among the vertices $T_1$, $F_1$, $T_2$, $F_2$, $\dots$, $T_n$, $F_n$, $L$. The order enforcer for the vertex $u$ connects $u$ to a newly added vertex $E_u$, adds a path $S_i$ of length two between $E_u$ and each special neighbor $u_i$ of $u$, and adds a V-gadget to the middle vertex of each $S_i$, see Figure~\ref{f:order}.
We attach an order enforcer simultaneously for each vertex $u\in \{T_2,F_2,T_3,F_3,\dots,T_n,F_n,L\}$. 

\begin{figure}
\begin{center}
\includegraphics[scale=1]{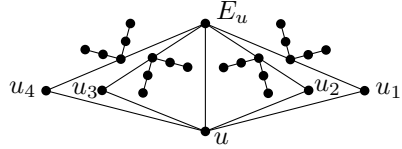}
\end{center}
\caption{Order enforcer for $u$ where $u_1, u_2, u_3, u_4$ are the special neighbors of $u$ and $E_u$ is the newly added vertex.}
\label{f:order}
\end{figure}

\subsubsection{The game}

In the following we make some easy observations.

\begin{observation}\label{obs_V}
In a game where Alice wins, no vertex of a V-gadget can be taken.
\end{observation}

\begin{proof}
Each V-gadget in $G_{\Phi}$ is connected to
other vertices only at its middle vertex $c$, see
Figures~\ref{f:var} and \ref{f:clause}. By rule (R),
the first vertex taken from the V-gadget is
$a$ or $a'$. Consequently, by rule (T), it can
be taken only in the first turn. If Alice takes vertex $a$ in the first
turn, then Bob has
to take vertex $b$. By the rules there is no further vertex that Alice
could take. So she loses the game. The case of $a'$ is analogous.
\end{proof}

As a straightforward consequence of Observation~\ref{obs_V} we get that no vertex $C_l$ corresponding to some clause $c_l$ can be taken from $G_{\Phi}$.

\begin{observation}\label{obs_noboth}
For every $i$, only one of $T_i$ and $F_i$ can be taken.
\end{observation}

\begin{proof}
Since $G_{\Phi}$ has at least two V-gadgets, Observation~\ref{obs_V} implies that taking both $T_i$ and $F_i$ would disconnect $G_{\Phi}$. 
\end{proof}

\begin{observation}\label{obs_enforcer}
 Let $u$ be a vertex in $G_{\Phi}$ to which an order enforcer is attached. If in the first turn Alice takes $u$ or $E_u$, then she loses the game.
\end{observation}

\begin{proof}
Suppose that Alice takes $u$. Then Bob takes $E_u$ and 
there is no further vertex for Alice to take, as each special neighbor of $u$ would disconnect the order enforcer and every other neighbor of $u$ or $E_u$ is the middle vertex of a V-gadget. If Alice takes $E_u$, then Bob takes $u$. Similarly, Alice cannot take any further vertex. In both cases Alice loses the game.
\end{proof}

Using similar arguments as in the previous proof we also observe the following fact.

\begin{observation}\label{obs_Eu}
If some neighbor of $u$ is taken, $E_u$ cannot be taken anymore. 
\qed
\end{observation}

The game must proceed as follows. As a consequence of Observations~\ref{obs_V},~\ref{obs_noboth} and~\ref{obs_enforcer} Alice will take $T_1$ or $F_1$ in the first turn. By Observations~\ref{obs_noboth} and~\ref{obs_Eu}, the only possible choices in the $i$th turn for the player on turn are $T_i$ and $F_i$ for $i\leq n$. In the $(n+1)$st turn Alice has to take $L$ or she has no turn and loses the game.

If Alice cannot take $L$, it means that some vertex $C_l$ corresponding to a clause $c_l$ would get disconnected from the part of $G_{\Phi}$ that contains the remaining vertices $T_i$ and $F_i$. This occurs if and only if previously all three vertices corresponding to the literals of $c_l$ were taken. If after taking $L$ the subgraph induced by the remaining vertices of $G_{\Phi}$ is connected, then there is no further vertex to take as it would necessarily disconnect $G_{\Phi}$. 

For $i\leq n$, let $x_i$ be TRUE if $T_i$ was taken by one of the players and FALSE if $F_i$ was taken. It follows that the formula $\varphi$ is satisfied if and only if Alice can take $L$ at the end of the game, that is, if and only if she wins the game. As Alice's turns in $G_{\Phi}$ correspond to the variables with the existential quantifier in $\Phi$ and Bob's turns in $G_{\Phi}$ to the variables with the universal quantifier, it follows that Alice wins if and only if $\Phi$ is a true formula.

Obviously the construction of $G_{\Phi}$ was carried out in polynomial time. This completes the proof of Theorem~\ref{veta_canonical_pspace} for the canonical game TR. 

\subsubsection{The misere game TR}
The proof of the PSPACE-completeness for the misere game TR is very similar to the proof for the canonical game TR, with the following differences. Instead of $n$ even we consider $n$ odd, so that Alice chooses between $T_n$ and $F_n$ and Bob then takes $L$ or has no move in the $(n+1)$st turn. The V-gadget now consists of a path of length $2$ attached by its middle vertex, and the order enforcer for a vertex $u$ is now doubled, with an additional edge $E_uF_u$. See Figures~\ref{obr_clause_misere} and~\ref{obr_enforcer_misere}.

\begin{figure}
\begin{center}
\includegraphics[scale=1]{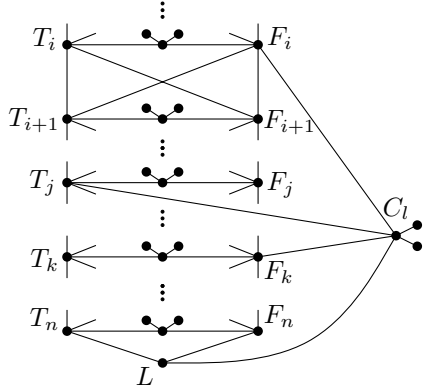}
\end{center}
\caption{A part of the graph $G_{\Phi}$ constructed for the misere game TR.}
\label{obr_clause_misere}
\end{figure}

\begin{figure}
\begin{center}
\includegraphics[scale=1]{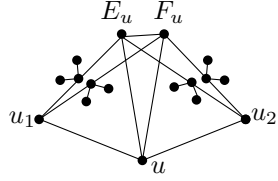}
\end{center}
\caption{An order enforcer for the misere game TR attached to a vertex $u$ and its two special neighbors.}
\label{obr_enforcer_misere}
\end{figure}


\subsection{Proof of Theorem~\ref{veta_weighted_pspace}}
As in the proof of Theorem~\ref{veta_canonical_pspace}, we show a polynomial reduction from the TQBF problem. Without loss of generality we may assume that the considered formula $\Phi$ is of the form $\Phi = \exists x_1 \forall x_2 \exists x_3 \forall x_4 \dots \exists x_n \varphi(x_1, x_2, \dots, x_n)$ where $\varphi$ is a NAE-$3$-SAT formula. That is, each clause of $\varphi$ has three literals and it is satisfied if and only if at least one of the three literals is evaluated as TRUE and at least one as FALSE.

\subsubsection{The construction of $G_{\Phi}$}

\begin{figure}
\begin{center}
\includegraphics[scale=1]{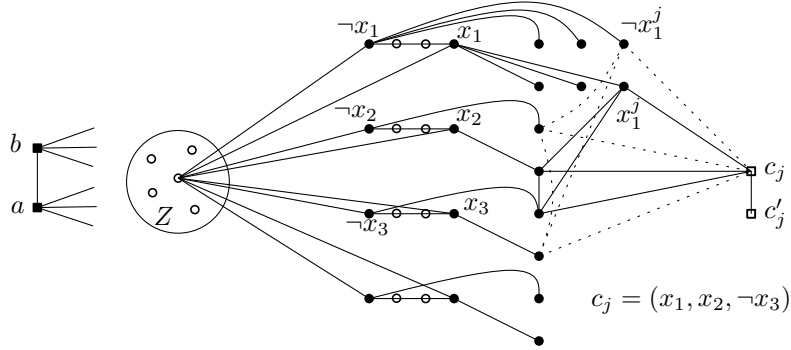}
\end{center}
\caption{An illustration of a part of the graph $G_{\Phi}$ for game R and TR.}
\label{f:weighted_1_graph}
\end{figure}

Let $m$ be the number of clauses in $\varphi$. We construct a $3$-connected weighted graph $G_{\Phi}$ of size $O(m+n)$ in the following way (see Figure~\ref{f:weighted_1_graph}). 
For each variable $x_i$ we take a path $P_i$ with 4
vertices. The end vertices of $P_i$ are labeled $x_i$ and $\neg x_i$.
If a variable $x_i$ occurs in the clause $c_j$, we add a vertex $x_i^j$ connected by an edge to $x_i$ and a vertex $\neg x_i^j$ connected by an edge to $\neg x_i$. The vertices $x_i^j$ and $\neg x_i^j$ are connected by clause gadgets depicted in Figure~\ref{f:weighted_2_clauses}. We add a set $Z$ of $10(m+n)+1$ 
vertices that are connected to all vertices $x_i$ and $\neg x_i$.
Finally, we add two vertices $a$ and $b$ connected to all other vertices (including the edge between $a$ and $b$). Observe that the constructed graph $G_{\Phi}$ is $3$-connected and has an odd number of vertices. 

\begin{figure}
\begin{center}
\includegraphics[scale=1]{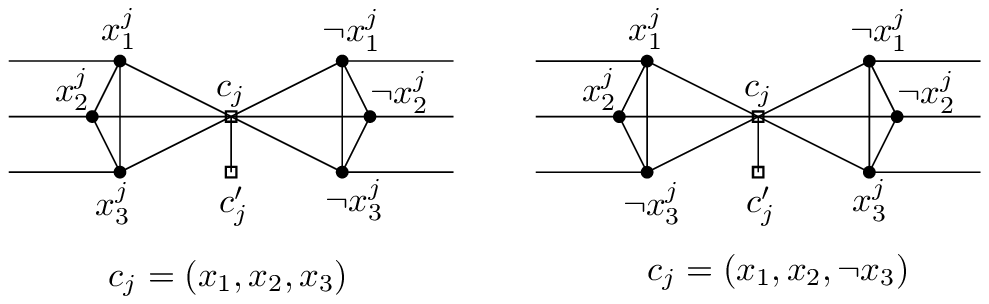}
\end{center}
\caption{Examples of NAE-3-SAT clause gadgets.}
\label{f:weighted_2_clauses}
\end{figure}

The weights $w:V(G_{\Phi}) \rightarrow [0, \infty)$ are set as follows: 

\begin{eqnarray*}
w(x_i)=w(\neg x_i) & = &9^{n+1-i},\\
w(x_i^j)=w(\neg x_i^j) & = & 10/999^j,\\
w(c_j)=w(c'_j) & = & 11/999^j,\\
w(b) & = & 9^{n+2},\\
w(a) & = & 9^{n+2}+2\cdot (9^{n-1}+9^{n-3}+ \cdots + 9^2) + 1/999^{m+1}.
\end{eqnarray*}

All other vertices (that is, the inner points on the paths $P_i$ and the vertices of $Z$) have weight $0$.
The $2n+8m+2$ vertices of positive weight can be partitioned into the following groups: $V_0=\{a,b\}, V_1=\{x_1, \neg x_1\}, \dots, V_n=\{x_n, \neg x_n\}$, and the $m$ groups $V_{n+1}, \dots, V_{n+m}$, each $V_{n+j}$ consisting of $8$ vertices of the clause gadget corresponding to the clause $c_j$. Note that the weights are chosen so that each vertex in a group $V_k$ has larger weight than the sum of weights of all the vertices from groups $V_l, l > k$. 
Let $W$ denote the total weight of all vertices in $G_{\Phi}$.

In the rest of this section we prove the following statement.

\begin{proposition}
Alice has a winning strategy in game TR played on $G_{\Phi}$ if and only if $\Phi$ is true.
Alice has a winning strategy in game R played on $G_{\Phi}$ if and only if $\Phi$ is true.
\end{proposition}

\subsubsection{Starting the game}

We start with the following easy observations.

\begin{observation}\label{obs_12_2}
Once one of the vertices $a$ or $b$ is taken, game TR reduces to game R as condition (T) is always satisfied further on. \qed
\end{observation}

We will call the graph induced by the remaining vertices in some position in the game briefly the {\em remaining graph}.

\begin{observation}\label{obs_cut_bob}
If in some position in game R played on $G_{\Phi}$ the remaining graph has a cut vertex $v$, then Bob has a strategy to get $v$.
\end{observation}

\begin{proof}
Let $H$ be the remaining graph after some of Alice's turns. If there is a component in $H - v$ containing at least two vertices, Bob takes a vertex from that component (at least one vertex from such component is always available). When no such component exists, the graph $H$ is a star with central vertex $v$ and odd number of leaves. Now Bob and Alice alternately take leaves of $H$, until only two vertices are left. Then $v$ becomes available for Bob.
\end{proof}

\begin{corollary}\label{cor_bob_any}
If Bob is on turn in game R, he has a strategy to get any of the remaining vertices.
\qed
\end{corollary}

%
In the analysis of the game we may without loss of generality assume that the players {\em play optimally\/}, that is, each of them follows a strategy maximizing his/her gain.

Now in a sequence of lemmas we show that the players do not have much freedom in choosing their strategy if they want to play optimally. Roughly speaking, the players will not deviate too much from a greedy strategy, which consists in taking an available vertex of largest weight.

\begin{lemma}\label{lemma_bob_posbira_zbyvajici_promenne}
Suppose that both players play optimally. Fix $i\in \{0,1,\dots,n\}$. Suppose further that Alice's first turn is on $a$, Bob's first turn is on $b$, and for each $j=1,2, \dots, i$ in the $(j+2)$nd turn of the game either $x_j$ or $\neg x_j$ is taken. Then Bob has a strategy to get all the remaining vertices from $V_1 \cup V_2 \cup \dots \cup V_i$.
\end{lemma}

\begin{proof}
For each $j=1,2, \dots, i$, let $w_j$ be the remaining vertex of the pair $x_j, \neg x_j$ and let $W_i=\{w_1, w_2, \dots, w_i\}$. Every vertex of $W_i$ is a cut vertex in the remaining graph and is adjacent to all vertices from $Z$. Each $w_j$ has a neighbor $y_j$ in the path $P_j$ that is still not taken. Let $Y_i=\{y_1,y_2,\dots,y_i\}$. Bob's strategy is to avoid taking vertices from $Y_i\cup Z$ as long as possible, and take a vertex of $W_i$ whenever it becomes available. When no such turn is possible, he takes an available vertex from $Z$.

First we observe that if the remaining graph $H$ after some of the Alice's following turns is not covered by $W_i\cup Y_i \cup Z$, then every component $C$ of $H-(W_i\cup Y_i \cup Z)$ contains a vertex that is not a cut vertex in $H$. Indeed, since $V(H)\cap (W_i\cup Y_i \cup Z)$ induces a connected subgraph of $H$, we may take the end vertex of the longest path starting in $V(H)\cap (W_i\cup Y_i \cup Z)$ and ending in $C$.


As the vertices of $Z$ are the only common neighbors of any pair of vertices from $W_i$, at most one vertex of $W_i$ is made available after each Alice's turn, unless Alice took the last remaining vertex of $Z$. But this will not happen as long as Bob is avoiding taking vertices from $Z$, since $Z$ contains more than half of the vertices of $G_{\Phi}$.

It follows that if Bob has no available vertex outside $Y_i \cup Z$, then all the vertices of the remaining graph belong to $W_i\cup Y_i \cup Z$. Hence in the rest of the game, Bob's strategy reduces to avoiding $Y_i$ and taking a vertex of $W_i$ whenever Alice makes it available. When only one vertex $z\in Z$ remains, we have a similar situation as in the proof of Theorem~\ref{veta_kconnected}: the remaining graph is a tree which is a union of paths $zw_jy_j$, it has an odd number of vertices and so Alice is on turn. Consequently Bob can get all the remaining vertices $w_j$. 
\end{proof}

\begin{lemma}\label{lemma_bob_b_x1_2x2}
Suppose that both players play optimally. Suppose further that Alice's first turn is on $a$ and Bob's first turn is on $b$. Then Bob has a strategy to get vertices of total weight at least $w(b)+w(x_1)+2w(x_2)$.
\end{lemma}

\begin{proof}
If Alice takes a vertex $v\notin\{x_1,\neg x_1\}$ in her second turn, then Bob takes $x_1$ in his second turn. This makes $\neg x_1$ unavailable for Alice's third turn. By Corollary~\ref{cor_bob_any}, Bob has a strategy to get $\neg x_1$ in the rest of the game. Therefore he gets vertices of total weight at least $w(b)+2w(x_1)>w(b)+w(x_1)+2w(x_2)$. 

Now suppose that Alice takes $x_1$ in her second turn (the case of $\neg x_1$ is symmetric). Then Bob takes $x_2$ in his second turn. By Lemma~\ref{lemma_bob_posbira_zbyvajici_promenne}, Bob has a strategy to take both $\neg x_1$ and $\neg x_2$ in the rest of the game.
\end{proof}

\begin{lemma}\label{lemma_ab}
If both players play optimally in game R or game TR, then Alice's first turn is on $a$ and Bob's first turn is on $b$.
\end{lemma}

\begin{proof}
Suppose that Alice takes $a$ in her first turn and then Bob takes $v \neq b$. 
If $v$ is not a cut vertex in $G_{\Phi} \setminus \{a,b\}$, then Alice can take $b$ in her second turn and get vertices of total weight at least $w(a)+w(b) > W-w(b)$ which implies that Bob's first turn was not optimal. If $v$ is a cut vertex in $G_{\Phi} \setminus \{a,b\}$ (namely, one of the vertices $c_j$), then $b$ becomes the only cut vertex in the remaining graph, so it is the only unavailable vertex for Alice in her second turn. Thus Alice can take $x_1$ in her second turn. As $b$ still separates the remaining vertices, Bob cannot take $b$ in his next turn.
It follows that Alice can take $\neg x_1$ or $x_2$ in her third turn, depending on Bob's choice in his second turn (we may assume without loss of generality that the formula $\Phi$ has at least two variables). In this way Alice gained vertices of total weight at least  $w(a) + 9^{n} + 9^{n-1} > W-w(b)-w(x_1)-2w(x_2)$. Therefore by Lemma~\ref{lemma_bob_b_x1_2x2}, Bob's first turn was not optimal.

Now suppose that Alice's first turn is on $u \neq a$. As both $a$ and $b$ are connected to all other vertices, Bob can play his first turn on $a$. By Observation~\ref{obs_12_2}, we can consider only game R from now on. 
If Alice took $b$ in her second turn, then she is not playing optimally, as she can get to a position with the same remaining graph using a better strategy: she takes $a$ in her first turn, then by the previous paragraph Bob takes $b$, and then Alice takes $u$. Thus Alice takes $v \neq b$ in her second turn. Then by Corollary~\ref{cor_bob_any}
Bob has a strategy to get $b$, so he gains at least $w(a)+w(b) > W-w(a)$. Hence Alice's only optimal first turn is on $a$.
\end{proof}

\subsubsection{Variable gadgets}

A {\em variable gadget\/} for the variable $x_i$ is the path $P_i$ connecting vertices $x_i$ and $\neg x_i$, with only the end vertices connected to $Z$. 

\begin{lemma}\label{lemma_xi}
If both players play optimally, then for each $i=1,2, \dots, n$, in the $(i+2)$nd turn of the game either $x_i$ or $\neg x_i$ is taken.
\end{lemma}

\begin{proof}
By Lemma~\ref{lemma_ab} and Observation~\ref{obs_12_2}, we need to consider only game R.

Let $i$ be the smallest positive integer that violates the statement of the lemma. 

Suppose first that the $(i+2)$nd turn was Alice's (thus $i$ is odd). Let $u \notin \{x_i, \neg x_i\}$ be the vertex taken by Alice in the $(i+2)$nd turn. If $u \in P_i$, then Bob can take the end vertex of $P_i$ adjacent to $u$ in the next turn. If $u \notin P_i$, then Bob takes $x_i$ in the next turn. 

We now argue that Bob has a strategy to get all the remaining vertices from $V_1\cup V_2\cup \dots \cup V_i$. 
We would get the same remaining graph if Alice and Bob switched the vertices they took in the $(i+2)$nd and the $(i+3)$rd turn. This modified game satisfies conditions of Lemma~\ref{lemma_bob_posbira_zbyvajici_promenne} until the $(i+2)$nd turn. Since Bob is now required to take $u$ in the $(i+3)$rd turn, we cannot apply Lemma~\ref{lemma_bob_posbira_zbyvajici_promenne} directly. But since $Z$ is large enough and taking $u$ does not make any remaining vertex of $V_1\cup V_2\cup \dots \cup V_i$ available to Alice, Bob may get all the remaining vertices from $V_1\cup V_2\cup \dots \cup V_i$ using the strategy from the proof of Lemma~\ref{lemma_bob_posbira_zbyvajici_promenne} from his next turn.

It follows that Bob has a strategy to get vertices of weight at least $w(b)+9^n+2\cdot 9^{n-1}+9^{n-2}+2\cdot 9^{n-3}+\cdots+2\cdot 9^{n-i+1}+2\cdot 9^{n-i} > W-w(a)-w(x_1)-w(x_3)-\cdots-w(x_i)$. Therefore the only optimal choice for Alice in the $(i+2)$nd turn was $x_i$ or $\neg x_i$.

Now suppose that $i$ is even. If Bob takes a vertex $v \notin \{x_i, \neg x_i\}$ in the $(i+2)$nd turn, then either $x_i$ or $\neg x_i$ is available for Alice in the next turn and she takes it. The total weight of the remaining vertices and the vertex $v$ is then smaller than $9^n + 9^{n-1}+ \cdots + 9^{n-i+2} + 2\cdot 9^{n-i+1}$. But by Lemma~\ref{lemma_bob_posbira_zbyvajici_promenne}, Bob has a strategy to get at least that much of the weight if he takes $x_i$ or $\neg x_i$ instead of $v$ in the $(i+2)$nd turn. This finishes the proof of Lemma~\ref{lemma_xi}.
\end{proof}

Let $w_1, w_2, \dots, w_n$ be the sequence of the remaining vertices from the groups $V_1, V_2, \dots, V_n$. The vertices $w_i$ determine a truth assignment $\sigma$ of the variables: $\sigma(x_i)={\rm TRUE}$ if $w_i=\neg x_i$ and $\sigma(x_i)={\rm FALSE}$ if $w_i=x_i$. See Figure~\ref{obr_weighted_3_evaluation}. For each $i=1,2,\dots,n$, let $y_i$ be the neighbor of $w_i$ in the path $P_i$.

\begin{figure}
\begin{center}
\includegraphics[scale=1]{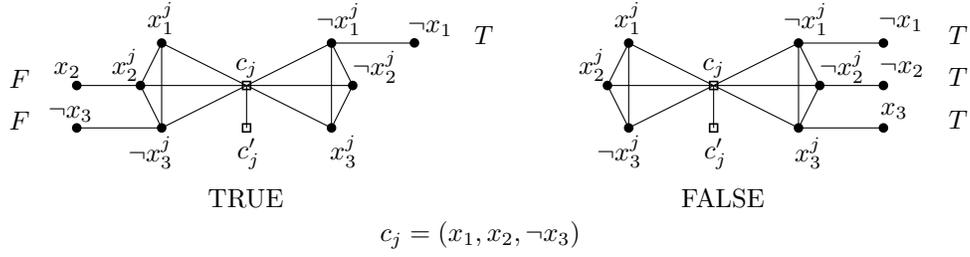}
\end{center}
\caption{A clause gadget after the players chose the evaluation $\sigma$ of the variables. Left: $\sigma(x_1)={\rm TRUE},\sigma(x_2)={\rm FALSE}, \sigma(x_3)={\rm TRUE}$, the clause $c_j$ is evaluated as TRUE. Right: $\sigma(x_1)={\rm TRUE},\sigma(x_2)={\rm TRUE}, \sigma(x_3)={\rm FALSE}$, the clause $c_j$ is evaluated as FALSE.}
\label{obr_weighted_3_evaluation}
\end{figure}

\subsubsection{Clause gadgets}

The subgraph induced by the group $V_{n+j}$ acts as a clause gadget for the clause $c_j$.
Call the two vertices of the group $V_{n+j}$ of weight $11/999^j$ {\em heavy\/} and the six vertices of weight $10/999^j$ {\em light\/}. 

\begin{lemma}\label{lemma_clause_gadget}
Suppose that all vertices $\{w_1, w_2, \dots, w_n\}$ and at least one vertex from $Z$ are still remaining and the players are restricted to play on the subgraph of $G_{\Phi}$ induced by $V_{n+j}$ with Bob being on turn. If both players play optimally, then Alice takes one heavy and three light vertices if the clause $c_j$ is satisfied by $\sigma$, otherwise she takes four light vertices. In other words, Alice takes exactly half of the weight of the group $V_{n+j}$ (namely $41/999^j$) if the clause $c_j$ is satisfied by $\sigma$ and $40/999^j$ otherwise.
\end{lemma}

\begin{proof}
Suppose that the clause $c_j$ is satisfied by $\sigma$. Then each of the triangles induced by light vertices is connected by an edge to some vertex $w_i$ (which is connected to all vertices in $Z$).
Bob can take a heavy vertex ($c'_j$) in his first turn. Alice can also take a heavy vertex in the next turn: either $c'_j$ if Bob took another vertex, or $c_j$ if Bob took $c'_j$. In the remaining turns each player takes three light vertices.

Now suppose that the clause $c_j$ is not satisfied by $\sigma$. Then Bob takes $c'_j$ in his first turn. The remaining heavy vertex $c_j$ is a separator in the remaining graph. Bob's strategy now is to take a light vertex from the same triangle as Alice. Once Alice takes the third light vertex of one of the two triangles, Bob takes $c_j$ and Alice thus gets no heavy vertex. 
\end{proof}


\begin{lemma}\label{lemma_vyherni_strategie}
Suppose that the players played optimally during the first $n+2$ turns.

{\rm (A)} If all clauses are satisfied after the first $n+2$ turns then Alice has a winning strategy for the rest of the game.

{\rm (B)} If at least one of the clauses is not satisfied after the first $n+2$ turns then Bob has a winning strategy for the rest of the game.
\end{lemma}

\begin{proof}
By the previous lemmas, the remaining vertices with positive weights are $w_1, w_2, \dots, w_n$, where $w_i\in V_i$, and all vertices from the groups $V_{n+1}$, $V_{n+2}$, $\dots$, $V_{n+m}$.

For every $j=1,2,\dots,m$, both Bob and Alice have a strategy to get four vertices of the clause gadget $V_{n+j}$:
if one of the players decides to take vertices from $V_{n+j}$ only, the other player cannot make all the remaining vertices of $V_{n+j}$ unavailable. Indeed, in such a case the remaining vertices would form a tree with every leaf connected to a unique component in the rest of the remaining graph. This is not possible as the only neighbors of $V_{n+j}$ are three vertices $w_i$ that are all connected to a large set $Z$.

\medskip
\noindent (A) We show that Alice wins if she follows the greedy strategy, that is, after the $(n+2)$nd turn she always takes an available vertex of the largest weight. The weights of the two vertices $a,b$ are set in such a way that in order to get more than one half of the total weight, Alice only needs to gather additional vertices of total weight $\sum_{j=1}^m 41/999^j$, which equals to the weight of a collection of one heavy and three light vertices from each clause gadget. Clearly, also any vertex $w_i$ is sufficient, so we can assume that Bob makes no $w_i$ available to Alice. Then Alice is sequentially taking vertices from the groups $V_{n+1}, V_{n+2}, \dots, V_{n+m}$, first the heavy ones and then the light ones. 

Let $i$ be the smallest integer such that after the $(n+2)$nd turn, Bob took a vertex outside $V_{n+i}$ when there was still at least one vertex available in $V_{n+i}$ and no vertex available in $V_{n+1}\cup V_{n+2} \cup \dots \cup V_{n+i-1}$. Then by Lemma~\ref{lemma_clause_gadget} and its proof, Alice got four vertices, including at least one heavy, from each of the groups $V_{n+1}, V_{n+2},\dots,V_{n+i-1}$. By the observation at the beginning of this proof, Alice has a strategy to get four vertices of $V_{n+i}$ as well. Again by the proof of Lemma~\ref{lemma_clause_gadget}, at least one of these vertices will be heavy. If Bob took two vertices outside $V_{n+i}$ while $V_{n+i}$ was still not completely taken, then Alice would get five vertices of $V_{n+i}$ and hence she would win the game. 

Therefore we can assume that Bob took only one vertex $v$ outside $V_{n+i}$ before all vertices from $V_{n+i}$ were taken. But then the position of the game just after taking the last vertex of $V_{n+i}$ is the same as in the case when Bob took $v$ later, in the $(n+3+8i)$th turn of the game (in his first turn after all vertices of $V_{n+i}$ were taken). So we may assume that Bob chose this latter strategy instead. By induction we conclude that Bob always took a vertex from the group $V_{n+i}$ with the smallest possible $i$, and hence Alice got one heavy and three light vertices from every clause gadget.

\medskip
\noindent (B)
Let $i$ be the smallest integer such that the clause $c_i$ is not satisfied. 
We show that Bob has a strategy to get all the vertices $w_j$ and vertices from $V_{n+1},V_{n+2},\dots,V_{n+i}$ of total weight at least $(\sum_{j=1}^{i-1} 41/999^j)+42/999^i$. This equals to the weight of one heavy and three light vertices from each of the groups $V_{n+1},V_{n+2},\dots,V_{n+i-1}$, and two heavy and two light vertices from $V_{n+i}$. Together with the vertices obtained during the first $n+2$ turns, this makes more than one half of the total weight $W$.

Bob starts with taking the vertex $c'_1$. Now we consider two phases of the remaining part of the game. The first phase lasts while in each of her turns, Alice takes a vertex from a group $V_{n+j}$ where $j\le i$ and $j$ is the smallest integer such that $V_{n+j}$ has still some vertex available. Bob's strategy in this phase is to take an available vertex from $V_{n+j}$ where $j$ is the smallest possible such integer, according to the strategy from the proof of Lemma~\ref{lemma_clause_gadget}. In this way Bob gets four vertices from each group $V_{n+j}$ that is completely taken; in particular, one heavy and three light vertices if the corresponding clause $c_j$ is satisfied, and two heavy and two light vertices if the clause $c_j$ is unsatisfied (that is, if $j=i$).

If all vertices from $V_{n+1}, V_{n+2},\dots, V_{n+i}$ are taken during the first phase, then Bob follows the following strategy in the second phase.
\begin{enumerate}
\item[1)] Take an available vertex $w_i$ of largest weight.
\item[2)] If no $w_i$ is available, take any vertex except the vertices $y_i$. 
\end{enumerate}
By the proof of Lemma~\ref{lemma_bob_posbira_zbyvajici_promenne}, Bob gets all the vertices $w_i$ in the rest of the game.

In the remaining case, let $k\le i$ be the smallest integer such that Alice takes a vertex $v$ outside of the group $V_{n+k}$ while some vertex in $V_{n+k}$ is still available and no vertex is available in $V_{n+1}\cup V_{n+2} \cup \dots \cup V_{n+k-1}$. Now Bob has to be a bit careful, since Alice can {\em expose\/} some of the vertices $w_j$ by taking its neighbor $y_j$ (after taking the other inner vertex $z_j$ of the path $P_j$). The neighbor of such an exposed vertex $w_j$ in $V_{n+k}$ may then become dangerous: a {\em dangerous\/} vertex is any available vertex which, after being taken, makes some of the vertices $w_j$ available. In particular, the vertex $y_j$ may become dangerous after its neighbor $z_j$ is taken. Note that a vertex may lose the property of being dangerous after several turns.

Bob's strategy is thus the following.
\begin{enumerate}
\item[1)] Take an available vertex $w_i$ of largest weight.
\item[2)] If no $w_i$ is available, take an available vertex from $V_{n+k}$ that is not dangerous.
\item[3)] If neither of the previous rules apply, take any vertex that is not dangerous. 
\end{enumerate}
We claim that using this strategy Bob gets five vertices of $V_{n+k}$ and all the vertices $w_1,w_2,\dots,w_n$, which is enough for him to win the game. 
 
First we observe that the only case when Bob can take no vertex from $V_{n+k}$ according to his strategy is when all the remaining vertices in $V_{n+k}$ are dangerous. Indeed, let $D$ be the set of dangerous vertices in $V_{n+k}$ and suppose that $V_{n+k}\setminus D$ is nonempty. Consider the longest path $P$ starting in $D$ with all other vertices in $V_{n+k}\setminus D$. Then the end vertex of $P$ in $V_{n+k}\setminus D$ is not a cut vertex in the remaining graph, since all dangerous vertices are connected through their neighbors $w_j$ and the set $Z$. 

Consider the first position in the game when all the remaining vertices of $V_{n+k}$ are dangerous. Let $d$ be the number of these dangerous vertices. For each dangerous vertex of $V_{n+k}$, Alice had to expose one vertex $w_j$, so she had to take two vertices $z_j$ and $y_j$ outside of $V_{n+k}$. It follows that Bob took at least $2d+1$ vertices from $V_{n+k}$ and thus $d\le 2$. If $d=2$, then Bob took at least $5$ vertices from $V_{n+k}$. If $d=1$, then Bob took at least $2$ more vertices than Alice from the $7$ taken vertices of $V_{n+k}$, which is at least $5$ again. Finally, if $d=0$, then Alice took precisely one vertex outside of $V_{n+k}$, so she took $3$ vertices from $V_{n+k}$ while Bob took $5$ of them.

Now by the proof of Lemma~\ref{lemma_bob_posbira_zbyvajici_promenne}, Bob gets all the vertices $w_i$ in the rest of the game. This finishes the proof.
\end{proof}

By the previous lemma, the existence of a winning strategy for Alice on the graph $G_{\Phi}$ is naturally reformulated as the quantified boolean formula $\Phi$, where the existential quantifiers correspond to Alice's turns and the universal quantifiers to Bob's turns on the variable gadgets.

\section*{Acknowledgments} We thank G\"unter Rote for inspiring us to consider also the misere game TR.


\end{document}